\documentclass[11pt,a4paper]{article}

\usepackage{fullpage}

\usepackage{xcolor}
\usepackage{bigints}
\usepackage{array}
\usepackage[normalem]{ulem}
\usepackage{xspace}
\usepackage{units}

\usepackage{microtype} 
\usepackage{amssymb,amsmath,amsthm}
\usepackage{hyperref}

\newcommand{\Reals}{\mathbb{R}}
\newcommand{\xis}{x^{\Join}}
\newcommand{\wis}{w_{i}^{\Join}}
\newcommand{\wio}{w_{i}^{\circ}}
\newcommand{\ddi}{\xi_i}
\newcommand{\ddone}{\xi_1}
\newcommand{\ddof}[1]{\xi_{#1}}
\newcommand{\xs}{x^{\ast}}
\newcommand{\ys}{y^{\ast}}
\newcommand{\vectori}{\left(\begin{matrix} x_{i} \\ y_{i} \\ 1
	\end{matrix}\right)}
\newcommand{\vectorii}{\left(\begin{matrix} x_{i+1} \\ y_{i+1} \\ 1
	\end{matrix}\right)}

\let\geq\geqslant
\let\leq\leqslant

\newtheorem{theorem}{Theorem}[section]
\newtheorem{lemma}{Lemma}[section]

\newbox\ProofSym
\setbox\ProofSym=\hbox{%
	\unitlength=0.18ex%
	\begin{picture}(10,10)
	\put(0,0){\framebox(9,9){}}
	\put(0,3){\framebox(6,6){}}
	\end{picture}}

\begin{document}
	
\title{Fitting a Graph to One-Dimensional Data\thanks{SWC is supported by Research Grants Council, Hong Kong, China
    (project no.~16203718). OC and TL were supported by ICT R\&D
    program of MSIP/IITP~[IITP-2015-0-00199].  A preliminary version appeared in Proceedings of the  Canadian Conference 
    on Computational Geometry, 2020~\cite{cccg}.}}

\author{Siu-Wing Cheng\footnote{Department~of~Computer~Science~and~Engineering,
    HKUST, Hong Kong. Email: {\tt scheng@cse.ust.hk, taegyoung.lee@connect.ust.hk, zrenah@connect.ust.hk}.}
  \and 
  Otfried Cheong\footnote{Korea Advanced Institute of Science \& Technology (KAIST).  Email: {\tt otfried@kaist.airpost.net}.}
  \and
  Taegyoung Lee\footnotemark[2]
  \and
  Zhengtong Ren\footnotemark[2]}

\maketitle
	
\begin{abstract}
  Given $n$ data points in $\mathbb{R}^d$, an appropriate
  edge-weighted graph connecting the data points finds application in
  solving clustering, classification, and regresssion problems.  The
  graph proposed by Daitch, Kelner and Spielman~(ICML~2009) can be
  computed by quadratic programming and hence in polynomial time.
  While a more efficient algorithm would be preferable,
  replacing quadratic programming is challenging even for the special
  case of points in one dimension.  We develop a dynamic programming
  algorithm for this case that runs in $O(n^2)$ time.  
\end{abstract}

\section{Introduction}

Many interesting data sets can be interpreted as point sets
in~$\Reals^{d}$, where the dimension~$d$ is the number of features of
interest of each data point, and the coordinates are the values of
each feature.  Given such a data set, graph-based semi-supervised learning is 
a paradigm for making predictions on the unlabelled data using the proximity among 
the data points and possibly some labelled data
(e.g.~\cite{daitch09,jebara09,liu10,ng01,xiang10,zhang14,zhou03}).  
Classification, regression, and clustering are some popular applications.  
The graph has to be set up first in order to perform the subsequent processing.  This 
requires the determination of the graph edges and the weights to be associated 
with the edges.  For example, let $w_{ij}$ denote the
weight determined for the edge that connects two points $p_i$ and
$p_j$, and regression can be performed to predict function values
$f_i$'s at the points $p_i$'s by minimizing $\sum_{i,j} w_{ij} (f_i -
f_j)^2$, subject to fixing the subset of known
$f_i$'s~\cite{daitch09}.   To allow efficient data
analysis, it is important that the weighted graph is sparse.


The graph connectivity should satisfy the property that similar discrete samples
are connected.  To this end, different proximity 
graphs have been suggested for connecting proximal points.  The
\emph{kNN}-graph connects each point to its $k$~nearest neighbors.
The $\varepsilon$-ball graph connects each point to all other points
that are within a distance~$\varepsilon$.  After fixing the graph
connectivity, edges to ``near'' points are given large weights and edges 
to ``far away'' points are given small weights.  That is, the larger the weight of an 
edge between points $p$ and $q$, the higher the influence of $q$ on $p$ and vice versa.  
It is thus inappropriate to
use the Euclidean distances among the points as edge weights.    Naively 
setting an edge weight as the reciprocal of the edge length
does not work either because the influence of a point is required to
fall much more rapidly as that point moves farther away.  It has been
proposed to associate a weight of $\exp(-\ell^2/2\sigma^2)$ to an edge of
Euclidean length $\ell$ for some \emph{a priori} determined parameter $\sigma$
(e.g.~\cite{ng01}).  
A well-tuned $\sigma$ is important.
A slight change in $\sigma$ 
may greatly affect the processing outcomes as observed in some previous 
work~(e.g.~\cite{xiang10}).  Several studies have found the \emph{kNN}-graph 
and the $\varepsilon$-ball graphs to be inferior to other proximity 
graphs~\cite{daitch09,han15,zhang14} for which both the graph connectivity
and the edge weights are determined simultaneously by solving an 
optimization problem.

We consider the graph proposed by Daitch, Kelner, and
Spielman~\cite{daitch09}.  It is provably sparse, and experiments have
shown that it offers good performance in classification, clustering
and regression.  This graph is defined via quadratic optimization as
follows: Let $P = \{p_1, p_2, \dots, p_n\}$ be a set of $n$~points
in~$\Reals^{d}$.  We assign weights $w_{ij} \geq 0$ to each pair of
points~$(p_i, p_j)$, such that $w_{ij} = w_{ji}$ and $w_{ii} = 0$.
These weights determine for each point~$p_i$ a vector~$\vec{v}_i$, as
follows:
\[
\vec{v}_i = \sum_{j = 1}^{n} w_{ij} (p_j - p_{i}).
\]
Let $v_i$ denote $\|\vec{v}_i\|$.  The weights are chosen so as to
minimize the sum
\[
Q = \sum_{i=1}^{n} v_{i}^{2},
\]
under the constraint that the weights for each point add up to at
least one (to prevent the trivial solution of $w_{ij} = 0$ for all $i$
and~$j$):
\[
\sum_{j=1}^{n} w_{ij} \geq 1 \qquad \text{for~$1 \leq i \leq n$}.
\]
The resulting graph contains an edge connecting $p_i$ and $p_j$ if and
only if $w_{ij} > 0$.

Daitch et al.~\cite{daitch09} showed that there is an optimal solution
where at most $(d+1)n$ weights are non-zero.  Moreover, in two
dimensions, optimal weights can be chosen such that the graph is
planar.

The optimal weights can be computed by quadratic
programming.  A quadratic programming problem with $m$ variables can be solved in $\tilde{O}(m^3)$ 
time in the worst case~\cite{monteiro89}.
In our case, there are $n(n-1)/2$ variables, which gives a worst-case running time of $\tilde{O}(n^6)$.  
Graphs based on optimizing other
convex quality measures have also been
considered~\cite{jebara09,zhang14}.  Our goal is to design an algorithm 
to compute the optimal weights in
Daitch et al.'s formulation that is significantly faster than
quadratic programming.  Perhaps surprisingly, this problem is
challenging even for points in one dimension, that is, when all points
lie on a line.  In this case, it is not difficult to show
(Lemma~\ref{lem:consecutive}) that there is an optimal solution such
that $w_{ij} > 0$ if and only if $p_i$ and $p_j$ are consecutive.

Despite its simplicity, the one-dimensional problem can model the task
of detecting change points and concept drift in a time series (e.g.~\cite{amin17,hido08,klink00,lee18,scholz07});
for example, seasonal changes in sales figures and customer behavior.  A time series of 
multi-dimensional data $(z_1,z_2,\cdots)$ is given,
and the problem is to decide the time steps $t$ at which there is a
``significant change'' from $z_{t-1}$ to $z_t$.   Suppose that
the ``distance'' between $z_{t-1}$ and $z_t$ can be computed according to some formula appropriate for the application (e.g.~\cite{hido08}).
By forming a path graph with vertices corresponding to the data points and edge weights determined as mentioned 
previously, one can apply clustering algorithms (e.g.~\cite{daitch09,ng01})  to group ``similar'' vertices and detect 
the change points as the boundaries of adjacent clusters.  This gives a potential application of the graph 
fitting problem in one dimension.

In general, although there are only $n-1$ variables in one dimension,
the weights in an optimal solution do not seem to follow any simple
pattern as we illustrate in the following two examples.

Some weights in an optimal solution can be arbitrarily high.  Consider
four points $p_1,p_2,p_3,p_4$ in left-to-right order such that $p_2
- p_1 = p_4 - p_3= 1$ and $p_3 - p_2 = \varepsilon$.  By
symmetry, $w_{12} = w_{34}$, and so $v_1 = v_4 = w_{12}$.  Since
$w_{12} + w_{23} \geq 1$ and $w_{23} + w_{34} \geq 1$ are trivially
satisfied by the requirement that $w_{12} = w_{34} \geq 1$, we can
make $v_2$ zero by setting $w_{23} = w_{12}/\varepsilon$.  In the optimal
solution, $w_{12} = w_{34} = 1$ and $w_{23} = 1/\varepsilon$.  So
$w_{23}$ can be arbitrarily large.

Given points $p_1, \cdots, p_n$ in left-to-right order, it seems ideal
to make $v_i$ a zero vector. One can do this for $i \in [2,n-1]$ by
setting $w_{i-1,i}/w_{i,i+1} = (p_{i+1}- p_i)/(p_i - p_{i-1})$,
however, some of the constraints $w_i + w_{i+1} \geq 1$ may be
violated.  Even if we are lucky that for $i \in [2,n-1]$, we can set
$w_{i-1,i}/w_{i,i+1} = (p_{i+1} - p_i)/(p_i - p_{i-1})$ without
violating $w_i + w_{i+1} \geq 1$, the solution may not be optimal as
we show below.  Requiring $v_i = 0$ for $i \in [2,n-1]$ gives $v_1 =
v_n = w_{12}(p_2 - p_1)$.  In general, we have $p_2 - p_1 \not=
p_n-p_{n-1}$, so we can assume that $p_2-p_1 >
p_n-p_{n-1}$.  Then, $w_{n-1,n} =
w_{12}(p_2-p_1)/(p_n-p_{n-1}) > 1$ as $w_{12} \geq 1$.  Since
$w_{n-1,n} > 1$, one can decrease $w_{n-1,n}$ by a small quantity
$\delta$ while keeping its value greater than 1.  Both constraints
$w_{n-1,n} \geq 1$ and $w_{n-2,n-1} + w_{n-1,n} \geq 1$ are still
satisfied.  Observe that $v_n$ drops to
$w_{12}(p_2-p_1)-\delta(p_n-p_{n-1})$ and $v_{n-1}$ increases to
$\delta(p_n-p_{n-1})$.  Hence, $v_{n-1}^2 + v_n^2$ decreases by
$2\delta w_{12}(p_2-p_1)(p_n-p_{n-1}) -
2\delta^2(p_n-p_{n-1})^2$, and so does~$Q$.  The original setting of
the weights is thus not optimal.  If $w_{n-3,n-2} + w_{n-2,n-1} > 1$,
it will bring further benefit to decrease $w_{n-2,n-1}$ slightly so
that $v_{n-1}$ decreases slightly from $\delta (p_n - p_{n-1})$ and
$v_{n-2}$ increases slightly from zero.  Intuitively, instead of
concentrating $w_{12}(p_2 - p_1)$ at $v_n$, it is better to
distribute it over multiple points in order to decrease the sum of
squares.  But it does not seem easy to determine the best weights.

Although there are only $n-1$ variables in one dimension, quadratic
programming still yields a running time of $\tilde{O}(n^3)$.  We
present a dynamic programming algorithm that computes the optimal
weights in $O(n^2)$ time in the one-dimensional case.  The
intermediate solutions computed by the algorithm have an interesting
structure where the derivative of the quality measure depends on the
derivative of a subproblem's quality measure as well as the inverse of
this derivative function.  We have experimental evidence that
explicitly representing the intermediate solutions may require
exponential time and space, and so we develop an implicit
representation that facilitates the dynamic programming algorithm.

\section{A single-parameter quality measure function}

We will assume that the points are given in sorted order, so that $p_1
< p_2 < p_3 < \dots < p_n$.  We first argue that the only weights that
need to be non-zero are the weights between consecutive points, that
is, weights of the form~$w_{i,i+1}$.

\begin{lemma}
  \label{lem:consecutive}
  For $d=1$, there is an optimal solution where only weights between
  consecutive points are non-zero.
\end{lemma}
\begin{proof}
  For a solution~$S$ and~$t \in \{1, 2, \dots, n-1\}$, let~$z_{t}(S)$
  be the number of non-zero weights at distance~$t$ (that
  is,~$z_{t}(S)$ is the cardinality of the set~$\{ i \mid w_{i,i+t} >
  0 \text{~in $S$}\}$), and let $\mathbf z(S)$ be the vector
  \[
  \mathbf z(S) = \big(z_{n-1}(S), z_{n-2}(S), \dots, z_3(S), z_{2}(S)\big).
  \]
  Among all optimal solutions, that is, solutions minimizing the
  quality~$Q$, let $O$ be a solution that lexicographically
  minimizes~$\mathbf z(O)$.  We will show that $\mathbf z(O) = (0,
  \dots, 0)$.

  Assume to the contrary that there is a weight~$w_{i,k} > 0$ for some
  $i < k-1$ in~$O$.  Let $j$ be an arbitrary index strictly between
  $i$ and $k$.  We construct a new optimal solution as follows: Let $a
  = p_{j} - p_{i}$, $b = p_{k} - p_{j}$, and $w = w_{i,k}$.  In the new
  solution, we set $w_{i,k}= 0$, increase $w_{i,j}$ by $\frac{a+b}{a}w$,
  and increase $w_{j,k}$ by~$\frac{a+b}{b}w$.  Note that since $a+b >
  a$ and $a+b > b$, the sum of weights at each vertex increases, and
  so the weight vector remains feasible.  The value $v_{j}$ changes by
  $-a \times \frac{a+b}{a} w + b \times \frac{a+b}{b} w = 0$; the
  value $v_{i}$ changes by $-(a+b)\times w + a \times \frac{a+b}{a} w
  = 0$; the value $v_{k}$ changes by $+(a+b)\times w - b \times
  \frac{a+b}{b} w = 0$.  It follows that the new solution~$O'$ has the
  same quality as~$O$, and is therefore also optimal.  Since $\mathbf
  z(O')$ is lexicographically smaller than~$\mathbf z(O)$, this is a
  contradiction to the choice of~$O$.
\end{proof}

To simplify the notation, we set $d_{i} = p_{i+1} - p_{i}$, for $1
\leq i < n$; rename the weights as $w_{i} := w_{i,i+1}$, again for
$1 \leq i < n$; observe that 
\begin{align*}
  v_{1} & = w_{1} d_{1}, \\
  v_{i} & = \left|w_{i} d_{i} - w_{i-1} d_{i-1}\right| \qquad \text{for $2 \leq i
    \leq n-1$}, \\
  v_{n} & = w_{n-1} d_{n-1}.
\end{align*}
For $i \in [2,n-1]$, we introduce the quantity
\begin{align*}
Q_{i} & = d_{i}^{2}w_{i}^{2} +
\sum_{j=1}^{i} v_{j}^{2} \\
& = d_{i}^{2}w_{i}^{2} + d_1^2w_1^2 +
\sum_{j=2}^{i} (d_jw_j - d_{j-1}w_{j-1})^2, 
\end{align*}
and note that $Q_{n-1} = \sum_{i=1}^n v_{i}^{2} = Q$.  Thus, our goal is to
choose the $n-1$ non-negative weights~$w_{1}, \dots, w_{n-1}$ such that $Q_{n-1}$
is minimized, under the constraints
\begin{align*}
  w_{1} & \geq 1, \\
  w_{j} + w_{j+1} & \geq 1 \qquad \text{for $2 \leq j \leq n-2$},\\
  w_{n-1} & \geq 1.
\end{align*}

The quantity~$Q_{i}$ depends on $w_{1}, w_{2}, \dots,
w_{i}$.  We concentrate on $w_i$ and consider
the function
\[
w_{i} \mapsto Q_{i}(w_{i}) = \min_{w_1, \dots, w_{i-1}} Q_{i}(w_1,
w_2, \dots, w_{i-1}, w_i),
\]
where the minimum is taken over all choices of $w_{1},\dots, w_{i-1}$
that respect the constraints $w_{1} \geq 1$ and $w_{j} + w_{j+1} \geq
1$ for $2 \leq j \leq i-1$.  The function $Q_{i}(w_{i})$ is defined
on~$[0, \infty)$.

We denote the derivative of the function~$w_{i} \mapsto Q_{i}(w_{i})$
by~$R_i$.  We will see shortly that~$R_{i}$ is a continuous, piecewise
linear function.  Since $R_{i}$ is not differentiable everywhere, we
define~$S_{i}(x)$ to be the right derivative of~$R_{i}$, that is
\[
S_{i}(x) = \lim_{~y \rightarrow x^+} R_i'(y).
\]
The following result discusses~$R_{i}$ and~$S_{i}$.  The shorthand
\[
\ddi := 2d_i d_{i+1}, \quad \mbox{for $1 \leq i < n-1$}, 
\]
will be convenient in
its proof and the rest of the paper.
\begin{theorem}
  \label{thm:ri}
  The function $R_i$ is strictly increasing, continuous, and piecewise
  linear on the range~$[0, \infty)$. We have $R_i(0) < 0$, $S_{i}(x)
    \geq (2 + \nicefrac 2i)d_{i}^{2}$ for all $x \geq 0$, and
    $R_{i}(x) = (2 + \nicefrac 2i)d_{i}^{2}x$ for sufficiently
    large~$x > 0$.
\end{theorem}
\begin{proof}
  We prove all claims by induction over~$i$.  The base case is
  $i=2$. Observe that
  \[
  Q_{2} = v_{1}^{2} + v_{2}^{2} + d_{2}^{2}w_{2}^{2}
  = 2d_{1}^{2}w_{1}^{2} - 2d_{1}d_{2}w_{1}w_{2} + 2d_{2}^{2}w_{2}^{2}.
  \]
  
  The derivative with respect to~$w_{1}$ is
  \begin{equation}
  \frac{\partial}{\partial w_{1}} Q_{2} =
  4d_{1}^{2}w_{1} - 2d_1 d_2 w_{2},
  \label{eq:0}
  \end{equation}
  which implies that $Q_{2}$ is minimized for $w_{1} =
  \frac{d_{2}}{2d_{1}}w_{2}$.  This choice is feasible (with respect
  to the constraint~$w_{1} \geq 1$) when $w_{2} \geq
  \frac{2d_1}{d_2}$.  If $w_{2} < \frac{2d_1}{d_2}$, then
  $\frac{\partial}{\partial w_1} Q_2$ is positive for all values
  of~$w_{1} \geq 1$, so the minimum occurs at $w_{1} = 1$.  It follows
  that
  \[
  Q_2(w_2) = \begin{cases}
    \frac 32 d_{2}^{2} w_{2}^{2} & \text{for } w_2 \geq \frac{2d_1}{d_2},
    \\
    2d_{2}^{2}w_{2}^{2} - \ddone w_2 + 2d_1^{2} & \text{otherwise},
  \end{cases}
  \]
  and so we have
  \begin{equation}
  R_2(w_2) = \begin{cases}
    3 d_{2}^{2} w_{2} & \text{for } w_2 \geq \frac{2d_1}{d_2},
    \\
    4d_{2}^{2}w_{2} - \ddone & \text{otherwise}.
  \end{cases}
  \label{eq:r2}
  \end{equation}
  In other words, $R_2$ is piecewise linear and has a single
  breakpoint at~$\frac{2d_1}{d_2}$.  The function $R_2$ is continuous
  because $3d_{2}^{2}w_2 = 4d_{2}^{2}w_{2}-\ddone$ when $w_{2} =
  \frac{2d_1}{d_2}$.  We have~$R_{2}(0) = -\ddone < 0$, $S_2(x) \geq
  3d_2^2$ for all~$x \geq 0$, and $R_{2}(x) = 3d_{2}^{2} x$ for $x
  \geq \frac{2d_1}{d_2}$.  The fact that $S_2(x) \geq 3d_2^2 > 0$
  makes $R_2$ strictly increasing.

  Consider now $i \geq 2$, assume that~$R_i$ and~$S_{i}$ satisfy the
  induction hypothesis, and consider~$Q_{i+1}$.  By definition, we
  have
  \begin{equation}
  Q_{i+1} = Q_{i} - \ddi w_{i}w_{i+1} + 2d_{i+1}^{2}w_{i+1}^{2}.
  \label{eq:1}
  \end{equation}
  For a given value of $w_{i+1} \geq 0$, we need to find the value of
  $w_{i}$ that will minimize~$Q_{i+1}$.  The derivative is
 \[
  \frac{\partial}{\partial w_{i}} Q_{i+1} =
  R_{i}(w_{i}) - \ddi w_{i+1}.
  \]
  The minimum thus occurs when $R_{i}(w_{i}) = \ddi w_{i+1}$.

  Since~$R_{i}$ is a strictly increasing continuous function with
  $R_{i}(0) < 0$ and $\lim_{x\rightarrow \infty} R_i(x)=\infty$, for
  any given $w_{i+1} \geq 0$, there exists a unique value~$w_{i} =
  R_{i}^{-1}(\ddi w_{i+1})$.  However, we also need to satisfy the
  constraint $w_{i} + w_{i+1} \geq 1$.

  We first show that $R_{i+1}$ is continuous and piecewise linear, and
  that $R_{i+1}(0) < 0$.  We will distinguish two cases, based on
  the value of $\wio := R_{i}^{-1}(0)$.
  
  \paragraph{Case~1:} $\wio \geq 1$.  This means that $R_{i}^{-1}(\ddi
  w_{i+1}) \geq 1$ for any $w_{i+1} \geq 0$, and so the constraint of
  $w_i + w_{i+1} \geq 1$ is satisfied for the optimal choice of~$w_{i}
  = R_i^{-1}(\xi_i w_{i+1})$.  It follows that
  \begin{align*}
    Q_{i+1}(w_{i+1}) & = Q_{i}\big(R_{i}^{-1}(\ddi w_{i+1})\big) - \ddi
    w_{i+1}R_i^{-1}(\ddi w_{i+1}) \\
    & \quad+ 2d_{i+1}^{2} w_{i+1}^{2}.
  \end{align*}
  The derivative $R_{i+1}$ is therefore
  \begin{align}
    R_{i+1}(w_{i+1}) & = R_{i}(R_{i}^{-1}(\ddi w_{i+1}))
    \frac{\ddi}{R'_{i}(R_{i}^{-1}(\ddi w_{i+1}))}
    \nonumber \\
    & \quad - \ddi R_{i}^{-1}(\ddi w_{i+1}) \nonumber \\
    & \quad - \ddi w_{i+1}
    \frac{\ddi}{R'_{i}(R_{i}^{-1}(\ddi w_{i+1}))} \nonumber \\
    & \quad + 4d_{i+1}^{2} w_{i+1} \nonumber \\
    & = 4d_{i+1}^{2} w_{i+1} - \ddi R_{i}^{-1}(\ddi w_{i+1}). \label{eq:2}
  \end{align}
  Since $R_{i}$ is continuous and piecewise linear, so is
  $R_{i}^{-1}$, and therefore~$R_{i+1}$ is continuous and piecewise
  linear.  We have $R_{i+1}(0) = -\ddi \wio < 0$.

  \paragraph{Case~2:} $\wio < 1$.  Consider the function $x \mapsto
  f(x) = x + R_{i}(x)/\ddi$.  Since $R_i$ is continuous and strictly
  increasing by the inductive assumption, so is the function $f$.
  Observe that $f(\wio) = \wio < 1$.  As $\wio < 1$, we have $R_i(1) >
  R_i(\wio) = 0$, which implies that $f(1) > 1$.  Thus, there exists a
  unique value~$\wis \in (\wio ,1)$ such that $f(\wis) = \wis +
  {R_{i}(\wis)}/{\ddi} = 1$.

  For $w_{i+1}\geq 1 - \wis = R_{i}(\wis)/\ddi$, we have
  $R_{i}^{-1}(\ddi w_{i+1}) \geq \wis$, and so $R_{i}^{-1}(\ddi
  w_{i+1}) + w_{i+1}\geq 1$.  This implies that the constraint $w_i + w_{i+1} \geq 1$ is satisfied when $Q_{i+1}(w_{i+1})$ is
  minimized for the optimal choice of $w_{i} = R_{i}^{-1}(\ddi w_{i+1})$.  So
  $R_{i+1}$ is as in~\eqref{eq:2} in Case~1.

  When $w_{i+1} < 1-\wis$, the constraint $w_{i} + w_{i+1} \geq 1$
  implies that $w_{i} \geq 1 - w_{i+1} > \wis$.  For any $w_{i} >
  \wis$ we have $\frac{\partial}{\partial w_{i}}Q_{i+1} = R_{i}(w_{i})
  - \ddi w_{i+1} > R_{i}(\wis) - \ddi (1-\wis) = 0$.  So $Q_{i+1}$ is
  increasing, and the minimal value is obtained for the smallest
  feasible choice of~$w_{i}$, that is, for $w_{i} = 1 - w_{i+1}$. It
  follows that
  \begin{align*}
    Q_{i+1}(w_{i+1}) & = Q_{i}(1-w_{i+1}) - \ddi w_{i+1}(1-w_{i+1}) \\
    & \quad + 2d_{i+1}^{2} w_{i+1}^{2} \\
    & = Q_{i}(1-w_{i+1}) - \ddi w_{i+1} \\
    & \quad + (\ddi + 2d_{i+1}^{2}) w_{i+1}^{2},
  \end{align*}
  and so the derivative $R_{i+1}$ is
  \begin{align}
    R_{i+1}(w_{i+1}) & = -R_{i}(1-w_{i+1}) \nonumber \\
    & \quad + (2\ddi + 4d_{i+1}^{2})w_{i+1} - \ddi.
    \label{eq:3}
  \end{align}
  
  Combining \eqref{eq:2} and \eqref{eq:3}, we have
  \begin{itemize}
  	\item If $w_{i+1} < 1-\wis$, then
  	\begin{align}
  	R_{i+1}(w_{i+1}) & = -R_{i}(1-w_{i+1}) \nonumber \\
  	& \quad + (2\ddi + 4d_{i+1}^{2})w_{i+1} - \ddi.
  	\label{eq:3-1}
  	\end{align}
  	\item If $w_{i+1} \geq 1-\wis$, then
  	\begin{align}
  	R_{i+1}(w_{i+1}) & = 4d_{i+1}^{2} w_{i+1} - \ddi R_{i}^{-1}(\ddi w_{i+1}).
  	\label{eq:3-2}
  	\end{align}
  \end{itemize}

  For $w_{i+1} = 1 - \wis$, we have $R_{i}(1-w_{i+1}) = R_{i}(\wis) =
  \ddi(1-\wis)$ and $R_{i}^{-1}(\ddi w_{i+1}) = R_{i}^{-1}(\ddi(1-\wis))
  = \wis$, and so both expressions have the same value:
  \begin{align*}
    & \quad -R_{i}(1-w_{i+1})  + (2\ddi + 4d_{i+1}^{2})w_{i+1} - \ddi \\
    & = \ddi\wis - \ddi + 2\ddi - 2\ddi\wis + 4d_{i+1}^{2}(1-\wis) - \ddi \\
    & = 4d_{i+1}^{2}(1-\wis) - \ddi\wis \\
    & = 4d_{i+1}^{2}(1-\wis) - \ddi R_{i}^{-1}(\ddi w_{i+1}).
  \end{align*}
  Since $R_{i}$ is continuous and piecewise linear, this implies that
  $R_{i+1}$ is continuous and piecewise linear.  We have $R_{i+1}(0) =
  -R_{i}(1) - \ddi$.  Since $\wio < 1$, we have $R_{i}(1) >
  R_{i}(\wio) = 0$, and so $R_{i+1}(0) < 0$.

  \medskip

  Next, we show that $S_{i+1}(x) \geq (2 + \nicefrac 2{i+1})
  d_{i+1}^2$ for all $x \geq 0$, which implies that $R_{i+1}$ is
  strictly increasing.  If $\wio < 1$ and $x < 1-\wis$, then by
  \eqref{eq:3-1},
  \begin{align*}
  S_{i+1}(x) & = S_{i}(1-x) + 2\ddi + 4d_{i+1}^{2} \\
  & \quad > 4d_{i+1}^{2} \\
  & \quad > (2 + \nicefrac 2{i+1})d_{i+1}^{2}.
  \end{align*}
  If $\wio \geq 1$ or $x > 1-\wis$, we have by~\eqref{eq:2}and \eqref{eq:3-2}
  that $R_{i+1}(x) = 4d_{i+1}^{2}x - \ddi
  R_{i}^{-1}(\ddi x)$.  By the inductive assumption that $S_i(x) \geq
  (2 + \nicefrac 2{i})d_i^2$ for all $x \geq 0$, we get
  $\frac{\partial}{\partial x}R_{i}^{-1}(x) \leq 1/\big((2 + \nicefrac
  2i)d_{i}^{2}\big)$.  It follows that
  \begin{align*}
    S_{i+1}(x) & \geq~4d_{i+1}^{2} - 
    \frac{(2d_{i}d_{i+1})^{2}}{(2+ \nicefrac 2i)d_{i}^{2}}~=~\Big(4 - \frac{4}{2 + \nicefrac 2i}\Big)d_{i+1}^{2}  \\
    & = \Big(4 - \frac{2i}{i + 1}\Big)d_{i+1}^{2} \\
    & = \Big(2 + \frac{2}{i+1}\Big) d_{i+1}^{2}.
  \end{align*}
  This establishes the lower bound on~$S_{i+1}(x)$.  
  
  Finally, by the inductive assumption, when $x$ is large enough, we have
  $R_{i}^{-1}(x) = x/\big((2 + \nicefrac 2i)d_{i}^{2}\big)$, and so
  \begin{align*}
    R_{i+1}(x) & = 4d_{i+1}^{2}x -
    \frac{(2d_{i}d_{i+1})^{2}}{(2+ \nicefrac 2i)d_{i}^{2}} x \\
    & = \Big(2 + \frac{2}{i+1}\Big) d_{i+1}^{2} x,
  \end{align*}
  completing the inductive step and therefore the proof.
\end{proof}

\section{The algorithm}

Our algorithm progressively constructs a representation of the
functions~$R_{2}, R_{3}, \dots, R_{n-1}$.  The function representation
supports the following three operations:
\begin{itemize}
\item Op~1: given $x$, return $R_{i}(x)$;
\item Op~2: given $y$, return $R_{i}^{-1}(y)$;
\item Op~3: given $\xi$, return $\xis$ such that $\xis +
  \frac{R_{i}(\xis)}{\xi} = 1$.
\end{itemize}

The proof of Theorem~\ref{thm:ri} gives the relation between $R_{i+1}$
and $R_i$.  This will allow us to construct the functions one by
one---we discuss the detailed implementation in
Sections~\ref{sec:basic}~and~\ref{sec:fast} below.

Once all functions~$R_2, \dots, R_{n-1}$ are constructed, the optimal
weights~$w_{1}, w_{2}, \dots, w_{n-1}$ are computed from the~$R_{i}$'s
as follows.  Recall that $Q = Q_{n-1}$, so $w_{n-1}$ is the value
minimizing~$Q_{n-1}(w_{n-1})$ under the constraint~$w_{n-1} \geq 1$.
If $R_{n-1}^{-1}(0) \geq 1$, then $R_{n-1}^{-1}(0)$ is the optimal
value for $w_{n-1}$; otherwise, we set $w_{n-1}$ to 1.

To obtain~$w_{n-2}$, recall from \eqref{eq:1} that $Q = Q_{n-1} = Q_{n-2}(w_{n-2}) -
\ddof{n-2}w_{n-2}w_{n-1} + 2d_{n-1}^{2}w_{n-1}^{2}$. Since we have
already determined the correct value of~$w_{n-1}$, it remains to
choose~$w_{n-2}$ so that $Q_{n-1}$ is minimized.  Since
\[
\frac{\partial}{\partial w_{n-2}}Q_{n-1} = R_{n-2}(w_{n-2}) -
\ddof{n-2}w_{n-1},
\]
$Q_{n-1}$ is minimized when $R_{n-2}(w_{n-2}) = \ddof{n-2}w_{n-1}$,
and so~$w_{n-2} = R_{n-2}^{-1}(\ddof{n-2}w_{n-1})$.

In general, for $i \in [2,n-2]$, we can obtain~$w_{i}$ from~$w_{i+1}$ by observing that 
\[
Q_{n-1} = Q_{i}(w_{i}) - \ddi w_{i}w_{i+1} + g(w_{i+1},\ldots,w_{n-1}),
\]
where $g$ is function that only depends on $w_{i+1}, \ldots, w_{n-1}$.
Taking the derivative again, we have
\[
\frac{\partial}{\partial w_{i}}Q_{n-1} = R_{i}(w_{i}) -
\ddi w_{i+1},
\]
so choosing $w_{i} = R_{i}^{-1}(\ddi w_{i+1})$ minimizes~$Q_{n-1}$.  To
also satisfy the constraint~$w_{i} + w_{i+1} \geq 1$, we need to
choose $w_{i} = \max\{R_{i}^{-1}(\ddi w_{i+1}), \, 1 - w_{i+1}\}$ for
$i \in [2,n-2]$.  Finally, from the discussion that immediately
follows~\eqref{eq:0}, we set $w_{1} =
\max\{\frac{d_{2}}{2d_{1}}w_{2},\, 1\}$.  To summarize, we have
\begin{align*}
  w_{n-1} & = \max\{R_{n-1}^{-1}(0), \, 1\},\\
  w_{i} &= \max\{R_{i}^{-1}(\ddi w_{i+1}), \, 1 - w_{i+1}\}, \,
  \text{for } i \in [2,n-2],\\
  w_{1} &= \textstyle\max\{\frac{d_{2}}{2d_{1}}w_{2},\, 1\}.
\end{align*}
It follows that we can obtain the optimal weights using a single Op~2
on each~$R_{i}$.

\subsection{Explicit representation of piecewise linear functions}
\label{sec:basic}

Since $R_{i}$ is a piecewise linear function, a natural representation
is a sequence of linear functions, together with the sequence of
breakpoints.  Since~$R_{i}$ is strictly increasing, all three
operations can then be implemented to run in time~$O(\log k)$ using
binary search, where $k$ is the number of function pieces.

We construct the functions~$R_i$, for $i = 2, \dots, n-1$, one by one.

The function~$R_{2}$ consists of exactly two pieces.  We construct it
directly from~$d_{1}, d_{2}$, and~$\xi_{1}$ using~(\ref{eq:r2}).

To construct~$R_{i+1}$, we make use of the explicit representation
of~$R_i$ that we have already computed.  We first compute~$\wio =
R_{i}^{-1}(0)$ using Op~2 on~$R_{i}$.  If $\wio \geq 1$, then by
\eqref{eq:2} each piece of~$R_{i}$, starting at the
$x$-coordinate~$\wio$, gives rise to a linear piece of~$R_{i+1}$, so
the number of pieces of~$R_{i+1}$ is at most that of~$R_{i}$.

If $\wio < 1$, then we compute $\wis$ using Op~3 on~$R_{i}$.  The new
function $R_{i+1}$ has a breakpoint at~$1-\wis$ by \eqref{eq:3-1} and
\eqref{eq:3-2}.  Its pieces for $x \geq 1-\wis$ are computed from the
pieces of~$R_{i}$ starting at the $x$-coordinate~$\wis$.  Its pieces
for $0 \leq x < 1-\wis$ are computed from the pieces of~$R_{i}$
between the $x$-coordinates~$1$ and~$\wis$. (Increasing~$w_{i+1}$ now
corresponds to a decreasing~$w_{i}$.)

Since a piece of~$R_{i}$ that covers $x$-coordinates in the
range~$[\wis, 1]$ gives rise to \emph{two} pieces of~$R_{i+1}$, the
number of pieces of~$R_{i+1}$ can be twice the number of pieces
of~$R_{i}$.  If this case occurs in every step of the construction,
then function~$R_{i}$ will have~$2^{i-1}$ pieces, leading to
exponential time and space complexity.

It is hard to manually construct inputs that exhibit exponential
growth in the number of pieces of~$R_{i}$.  We have therefore used the
Z3 theorem prover~\cite{z3} to try and satisfy the constraints under
which~$R_{n-1}$ has $2^{n-2}$~pieces. For $n = 10$, which is the
largest instance we have been able to solve using Z3, we obtain the
following nine distances:
\[
\begin{matrix}
  d_1 & d_2 & d_3 & d_4 & d_5 & d_6 & d_7 & d_8 & d_9 \\
  11569 & 49184 & 65536 & 98304 & 109056 & 145408 & 146432 & 147456 & 32768
\end{matrix}
\]
We have implemented the algorithm and verified using exact arithmetic
that for these distances, the function~$R_{i}$ indeed has~$2^{i-1}$
pieces, for $2 \leq i \leq 9$.

Based on this experiment, we find it unlikely that the explicit
representation of the functions~$R_{i}$ will lead to a polynomial-time
algorithm. 

\subsection{A quadratic time implementation}
\label{sec:fast}

To guarantee a polynomial running time, we turn to an implicit
representation of~$R_{i}$.  This representation is based on the fact
that there is a linear relationship between points on the graphs of
the functions~$R_{i}$ and~$R_{i+1}$.  Concretely, let $y_{i} =
R_{i}(x_{i})$, and $y_{i+1} = R_{i+1}(x_{i+1})$.  Recall the following
relation from \eqref{eq:2} for the case of $\wio \geq 1$:
\begin{align*}
  R_{i+1}(w_{i+1}) & = 4d_{i+1}^{2} w_{i+1} - \ddi R_{i}^{-1}(\ddi w_{i+1}).
\end{align*}
We can express this relation as a system of two equations:
\begin{align*}
  y_{i+1} & = 4d_{i+1}^{2} x_{i+1} - \ddi x_{i}, \\
  y_{i} & = \ddi x_{i+1}.
\end{align*}
This can be rewritten as
\begin{align*}
  y_{i+1} & = 4d_{i+1}^{2} y_{i} / \ddi - \ddi x_{i}, \\
  x_{i+1} & = y_{i} / \ddi,
\end{align*}
or in matrix notation
\begin{align}
  \vectorii & = M_{i+1} \cdot \vectori,
    \label{eq:10}
\end{align}
where
\begin{align*}
  M_{i+1} & = \left( \begin{matrix}
  0 & 1/\ddi & 0 \\
  -\ddi & 4d_{i+1}^{2}/\ddi & 0 \\
  0 & 0 & 1
  \end{matrix}\right).
\end{align*}

On the other hand, if $\wio < 1$, then $R_{i+1}$ has a breakpoint at
$1-\wis$. The value $\wis$ can be obtained by appying Op~3 to $R_i$.
We compute the coordinates of this breakpoint: $(1-\wis,
R_{i+1}(1-\wis))$.  Note that $R_{i+1}(1-\wis) = 4d_{i+1}^2(1-\wis) -
\xi_i R_i^{-1}(\xi_i(1-\wis))$ which can be computed by applying Op~2
to $R_i$.  For $x_{i+1} > 1-\wis$, the relationship between $(x_{i},
y_{i})$ and $(x_{i+1}, y_{i+1})$ is given by \eqref{eq:10}.  For $0
\leq x_{i+1} < 1-\wis$, recall from \eqref{eq:3} that
\begin{align*}
  R_{i+1}(w_{i+1}) & = -R_{i}(1-w_{i+1}) \\
  & \quad + (2\ddi +
  4d_{i+1}^{2})w_{i+1} - \ddi.
\end{align*}
We again rewrite this as
\begin{align*}
  y_{i+1} & = -y_{i} + (2\ddi + 4d_{i+1}^{2})x_{i+1} - \ddi, \\
  x_{i} & = 1 - x_{i+1},
\end{align*}
which gives
\begin{align*}
  y_{i+1} & = -y_{i} + (2\ddi + 4d_{i+1}^{2})(1- x_{i}) - \ddi, \\
  x_{i+1} & = 1 - x_{i},
\end{align*}
or in matrix notation:
\begin{align*}
  \vectorii & = L_{i+1} \cdot \vectori,
 \end{align*}
where
\begin{align*}
  L_{i+1} & = \left( \begin{matrix}
  -1 & 0 & 1 \\
  -2\ddi -4d_{i+1}^{2} & -1 & \ddi + 4d_{i+1}^{2} \\
  0 & 0 & 1
  \end{matrix}\right).
\end{align*}

We will make use of this relationship to store the function~$R_{i+1}$,
for $i \geq 2$, by storing the breakpoint~$(\xs_{i+1}, \ys_{i+1}) =
(1-\wis,R_{i+1}(1-\wis))$ as well as the two matrices~$L_{i+1}$
and~$M_{i+1}$.  The function~$R_2$ is simply stored explicitly.

We now discuss how the three operations Op~1, Op~2, and Op~3 are
implemented on this representation of a function~$R_{i}$.  For an
operation on~$R_{i}$, we progressively build transformation matrices
$T_{i}^{i}, T_{i-1}^{i}, T_{i-2}^{i}, \dots, T_{3}^{i}, T_{2}^{i}$
such that $(x_{i}, y_{i}, 1) = T_{j}^{i} \times (x_{j}, y_{j}, 1)$ for
every~$2 \leq j \leq i$ in a neighborhood of the query.  Once we
obtain~$T_{2}^{i}$, we use our explicit representation of~$R_{2}$ to
express~$y_{i}$ as a linear function of~$x_{i}$ in a neighborhood of
the query, which then allows us to answer the query.

The first matrix~$T_{i}^{i}$ is the identity matrix. We obtain
$T_{j}^{i}$ from~$T_{j+1}^{i}$, for $j \in [2,i-1]$, as follows:
If~$R_{j+1}$ has no breakpoint, then $T_{j}^{i} = T_{j+1}^{i} \cdot
M_{j+1}$.  If $R_{j+1}$ has a breakpoint~$(\xs_{j+1}, \ys_{j+1})$,
then either $T_{j}^{i} = T_{j+1}^{i}\cdot M_{j+1}$ or $T_{j}^{i} =
T_{j+1}^{i}\cdot L_{j+1}$, depending on which side of the breakpoint
applies to the answer of the query.  We can decide this by comparing
$(x', y', 1)^t = T_{j+1}^{i} \cdot (\xs_{j+1}, \ys_{j+1}, 1)^t$ with the
query.  More precisely, for Op~1 we compare the input~$x$ with~$x'$,
for Op~2 we compare the input~$y$ with~$y'$, and for Op~3 we compute
$x' + y'/\xi$ and compare with~$1$.

Assuming the Real-RAM model common in computational geometry, where
arithmetic on real numbers takes constant time, it follows that the
implicit representation of~$R_{i}$ supports all three operations
on~$R_{i}$ in time~$O(i)$.

Finally, we discuss how the representation of all functions~$R_i$ is
obtained.  We again build it iteratively, constructing~$R_{2}, R_3,
R_4, \dots, R_{n-1}$, one-by-one in this order.  The first
function~$R_{2}$ is stored explicitly.  To construct the implicit
representation of~$R_{i+1}$, we only need to perform on our representation of~$R_{i}$ (that we already computed)
one Op~2 to get $\wio = R_i^{-1}(0)$, one Op~3 to get $\wis$, and one Op~2 to get $R_i^{-1}(\xi_i(1-\wis))$, which allows us to determine the breakpoint $(1-\wis,R_{i+1}(1-\wis))$, if there is one (when $\wio < 1$).  The two matrices~$L_{i+1}$
and~$R_{i+1}$ can be computed in $O(1)$ time.

Since operations on~$R_{i}$ take time~$O(i)$, the total time to
construct~$R_{n-1}$ is~$O(n^{2})$.
\begin{theorem}
  \label{thm:alg} Given $n$ points on a line, we can compute an
  optimal set of weights for minimizing the quality measure $Q$ in
  $O(n^2)$ time under the Real-RAM model.
\end{theorem}

\section{Conclusion}

We do not have a polynomial time bound on the running time using the
explicit representation of the functions~$R_i$, and our experiments
with the Z3~solver suggest that the complexity of~$R_i$ may indeed
increase exponentially with~$i$.  It would be interesting to determine
whether one can construct such a worst-case example for any~$n$.

It would also be nice to obtain an algorithm for higher dimensions
that is not based on a quadratic programming solver.

In two dimensions, we have conducted some experiments that indicate
that the Delaunay triangulation of the point set contains a
well-fitting graph.  If we choose the graph edges only from the
Delaunay edges and compute the optimal edge weights, the resulting
quality measure is very close to the best quality measure in the
unrestricted case.  It is conceivable that one can obtain a provably
good approximation from the Delaunay triangulation.

\small
\bibliographystyle{abbrv}

\end{document}